\documentclass[letterpaper, 10 pt, conference]{ieeeconf}  %
\IEEEoverridecommandlockouts  
\usepackage[utf8]{inputenc}
\usepackage{xcolor}
\usepackage{graphicx}
\usepackage{graphics}
\usepackage{subcaption}
\usepackage{amsmath}
\usepackage[version=4]{mhchem}
\usepackage{siunitx}
\usepackage{longtable,tabularx}
\usepackage{algorithm}
\usepackage{algpseudocode}
\usepackage{amsmath}
\usepackage{amssymb}
\usepackage{proof}
\usepackage{epsfig}
\usepackage{calc}

\usepackage{accents}
\newcommand{\dbtilde}[1]{\accentset{\approx}{#1}}

\newtheorem{prop}{{{{Proposition}}}}

\newtheorem{remark}{Remark}
\newtheorem{theorem}{Theorem}
 
\usepackage[autostyle]{csquotes}

\IEEEoverridecommandlockouts 

\renewcommand{\baselinestretch}{.95}

\title{\LARGE \bf
Koopman Operator Based Modeling for Quadrotor Control on $SE(3)$
}

\author{Vrushabh Zinage \and Efstathios Bakolas
\thanks{Vrushabh Zinage (graduate student) and Efstathios Bakolas (Associate Professor) are with the Department of Aerospace Engineering and Engineering Mechanics,
The University of Texas at Austin, Austin, Texas 78712-1221, USA, {Emails:\tt\small vrushabh.zinage@utexas.edu; bakolas@austin.utexas.edu}}
\thanks{This research has been supported  in part by NSF award CMMI-1937957.}
}

\begin{document}

\bibliographystyle{IEEEtran} 

\maketitle
\thispagestyle{empty}
\pagestyle{empty}

\begin{abstract}
In this paper, we propose a Koopman operator based approach to describe the nonlinear dynamics of a quadrotor on $SE(3)$ in terms of an infinite-dimensional linear system which evolves in the space of observable functions (lifted space) and which is more appropriate for control design purposes. The major challenge when using the Koopman operator is the characterization of a set of observable functions that can span the lifted space. Recent methods either use tools from machine learning to learn the observable functions or guess a suitable set of observables that best describes the nonlinear dynamics. Instead of guessing or learning the observables, in this work we derive them in a systematic way for the quadrotor dynamics on $SE(3)$. In addition, we prove that the proposed sequence of observable functions converges pointwise to the zero function, which allows us to select only a finite set of observable functions to form (an approximation of) the lifted space. Our theoretical analysis is also confirmed by numerical simulations which demonstrate that by increasing the dimension of the lifted space, the derived linear state space model can approximate the nonlinear quadrotor dynamics more accurately. 
\end{abstract}

\section{INTRODUCTION}
We consider the problem of forming a \textit{lifted} space over which the nonlinear dynamics of a quadrotor on $SE(3)$ can be described by a (possibly infinite-dimensional) linear system. The approach utilized herein, which relies on the framework of Koopman operator, allows one to account for the nonlinearities of the dynamics of the quadrotor while at the same time linear control design techniques are still applicable. The major challenge in using the Koopman operator for approximating nonlinear dynamics is finding a suitable set of \textit{observable} functions (or \textit{observables}) that can serve as basis functions for the lifted space. In this paper, we propose a systematic way to derive a set of observable functions for the quadrotor dynamics on $SE(3)$. A subset of the latter set is chosen to form the truncated (approximation of the) \enquote{lifted} linear state space model for the quadrotor dynamics. In prior literature \cite{mellinger2012trajectory_quad_dyn_1,hoffmann2008quadrotor_quad_dyn_2}, control design methods for quadrotors are based on linearization of their dynamics around a reference trajectory or a fixed point to facilitate the use of linear control design tools such as linear Model Predictive Control (MPC) \cite{bangura2014real_mpc_1}. These methods achieve reduced computational overhead compared to Nonlinear MPC (NMPC), thereby allowing more on-board processing power for other applications like communication and perception. However, the desired accuracy from linearized based models cannot be guaranteed for large  deviations from the linearization point. Other more sophisticated linearization methods like Carleman linearization \cite{kowalski1991nonlinear_carleman} cannot be used for quadrotors as their applicability is limited to polynomial vector fields.

Quadrotor systems have received a lot of attention in the robotics and control communities. The strong nonlinear coupling and underactuated characteristics of quadrotor dynamics make the design of an effective controller more difficult than other mobile robots. In this paper, we use a {Koopman operator based approach to derive a set of observable functions which will allow us to approximate the nonlinear dynamics of a quadrotor on $SE(3)$ by a higher dimensional linear state space model which is more appropriate for (model based) control design purposes}.

{\textit{Literature review:}} Koopman operator-based methods have been widely used for approximating nonlinear systems \cite{budivsic2012applied_koopman_1,mauroy2016global_koopman_2,mezic2005spectral_koopman_3,surana2018koopman_koopman_4,klus2018data_koopman_5,susuki2011nonlinear_koopman_6}. Extensions of these methods for controlled systems have been proposed in \cite{korda2018linear_koopman_mpc_automatica,proctor2016dynamic_koopman_controls_1,kaiser2020data_koopman_controls_2,far_rendezvous_zinage_bakolas_acc_2021,you2018deep_koopman_controls_3,ma2019optimal_koopman_controls_4}. Koopman-based approaches have been proposed for robotic applications in~\cite{abraham2019active_data_quad_dyn,bruder2020data_soft_robots}. The main challenge in using Koopman operator based methods is choosing the right set of observable functions for modeling purposes \cite{brunton2016koopman_observable_limitation}.  Approaches using neural networks \cite{lusch2018deep_nn_1,yeung2019learning_deep_nn_2}, reproducing kernels \cite{kawahara2016dynamic_reproducing_kernels}, and basis functions \cite{williams2015data_basis_function} have also been proposed for estimating the Koopman operator. However, these methods require prior data and a suitable finite set of observable functions which are typically chosen in a heuristic way. In addition, there is no guarantee that the data collected is rich enough for estimating the Koopman operator that can approximate the nonlinear dynamics accurately enough. \cite{netto2020analytical} proposes a method for analytically finding the observables for nonlinear dynamics which can be described as a linear combination of elementary functions like sine, cosine and exponential.
In \cite{chen2020koopman_attitude}, a set of observable functions are derived and thereafter a Koopman based LQR controller is proposed for the spacecraft attitude dynamics.

{\textit{Main contributions:}} In contrast with \cite{chen2020koopman_attitude}, which has inspired this work, we derive a set of observable functions for the quadrotor dynamics on $SE(3)$ in which the attitude and position of the quadrotor are coupled in a nonlinear fashion. Subsequently, we use these functions to form the lifted-space in which the quadrotor dynamics is described by a linear state space model. {This linear model is equivalent to the nonlinear model as the dimension of the lifted space tends to infinity.} In addition, we prove the pointwise convergence of the observable functions to the zero function as the dimension of the lifted state space goes to infinity. This allows us to truncate the higher dimensional lifted space to a lower (finite) dimensional lifted space.
The main advantage of using the derived set of observable functions is that it does not require prior data of state and control input pairs for using the Koopman operator based control approaches. To the best of our knowledge, this is the first paper which derives a set of observable functions for both position and attitude dynamics of a quadrotor on $SE(3)$.

{\textit{Structure of the paper:}} The organization of the paper is as follows. In Section \ref{sec:prelimanaries}, we introduce the nonlinear state space model for a quadrotor followed by an overview of the Koopman operator. In Section \ref{sec:derivation_of_observable_functions}, we provide 
the derivation of the set of observable functions and then the formulation of the lifted space linear dynamics. Numerical simulations are presented in Section \ref{sec:results}, and finally Section \ref{sec:conclusions} presents concluding remarks.

\section{Preliminaries\label{sec:prelimanaries}}
\subsection{Nomenclature\label{subsection:nomenclature}}
Given a vector $\boldsymbol{a}\in\mathbb{R}^n$, let $(\boldsymbol{a})_i$ denote its $i^\text{th}$ element and let $|\boldsymbol{a}|$ denote its Euclidean norm. Given $\boldsymbol{a}=[a_1,\;a_2,\;a_3]^\mathrm{T}\in\mathbb{R}^3$, we denote by ${a}^\times\in\mathbb{R}^{3\times 3}$ the matrix which is such that the cross product $\boldsymbol{a}\times\boldsymbol{b}={a}^{\times}\boldsymbol{b}$ for all $\boldsymbol{b}\in\mathbb{R}^3$. The set of natural numbers is denoted by $\mathbb{N}$. {Given $a,b\in \mathbb{N}$ with $a\leq b$, we denote the discrete interval from $a$ to $b$ as $[a,b]_d$, where $[a,b]_d:=[a,b] \cap \mathbb{N}$}. Let $\mathrm{bdiag}(D_1,\dots,D_k)$ denote the block diagonal matrix comprising of matrices $D_1,\dots,D_k$. Let $\mathbf{0}_{n \times m}$ (or simply $\mathbf{0}$) denote the $n\times m$ zero matrix, $\mathbf{0}_{n}$ denote the $n\times n$ zero matrix and ${I}_n$ denote the $n\times n$ identity matrix. In addition, we denote by $\underline{A}=\texttt{vec}(A)$, the vector that is formed by concatenating the columns that comprise the matrix $A$. The trace operator is denoted as $\texttt{tr}(\cdot)$. Let $\|X\|_F$ denotes the Frobenius norm for matrix $X\in\mathbb{R}^{n\times m}$ where $\|X\|_F:=|\texttt{vec}(X)|$. For a matrix $X$, we denote by $X^{\dagger}$ its Moore-Penrose inverse. In addition, let $\otimes$ denote the Kronecker product. {By $SE(3)$, we denote the Special Euclidean group which can be represented as follows:}
$
{SE(3)=\left\{\mathbf{A} \mid \mathbf{A}=\left[\begin{smallmatrix}
{R} & \boldsymbol{p} \\
\mathbf{0} & 1
\end{smallmatrix}\right],~\boldsymbol{p} \in \mathbb{R}^3,~RR^\mathrm{T}={R}^\mathrm{T} {R}={I}_3\right\}}$. {Let $J\in\mathbb{R}^{3\times 3}$ denote the quadrotor's inertia (positive definite) matrix, $\boldsymbol{p}\in\mathbb{R}^{3}$ its position in the inertial frame, and $R\in \mathrm{SO}(3)$ its rotation matrix from body-fixed frame to inertial frame (see Fig. \ref{fig:quadrotor}). Finally, let $\boldsymbol{v}\in\mathbb{R}^3$ and $\boldsymbol{\omega}\in\mathbb{R}^3$ denote the linear and angular velocities of the quadrotor in the body-fixed frame, respectively, and $m$ its mass. }
\subsection{Quadrotor dynamics\label{sec:quadrotor_dynamics}}
\begin{figure}[h]
\centering
\includegraphics[width=4cm]{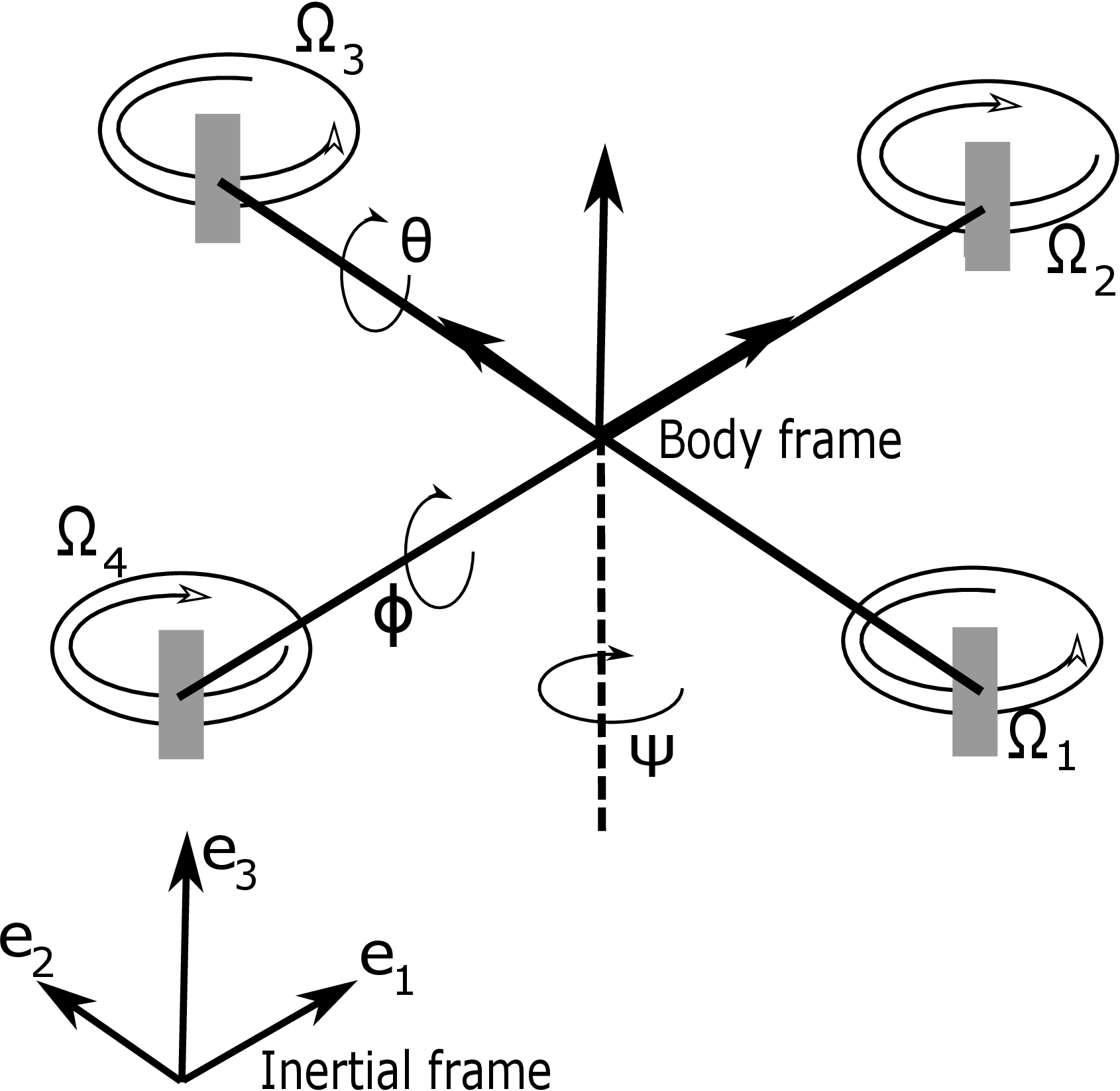}
\caption{Inertial and body fixed frame of the quadrotor}
\label{fig:quadrotor}
\end{figure}

In this section, the dynamics of the quadrotor are discussed. 
The quadrotor is an under-actuated system whose motion can be described as follows~\cite{abraham2019active_data_quad_dyn,lee2010geometric}:
\begin{align}
& \dot{h}=h\left[\begin{smallmatrix}
{\omega}^\times & \boldsymbol{v} \\
\mathbf{0} & 0
\end{smallmatrix}\right],\;J \dot{\boldsymbol{\omega}} =M+J \omega^\times \boldsymbol{\omega},\nonumber\\
& \dot{\boldsymbol{v}} =({1}/{m}) F \boldsymbol{e}_{3}-\omega^\times \boldsymbol{v}-\boldsymbol{g} R^\mathrm{T} \boldsymbol{e}_{3},
\label{eqn:nonlinear_quad_dynamics}
\end{align}
where $h\in SE(3)$, 
$g$ is the acceleration due to gravity and $\boldsymbol{e}_3=[0,0,1]^\mathrm{T}$.
Finally, 
\begin{align}
&F=k_{t}\left(u_{1}+u_{2}+u_{3}+u_{4}\right),\nonumber \\
&M=[k_{t} l(u_{2}-u_{4}), k_{t} l(u_{3}-u_{1}),  
k_{m}(u_{1}-u_{2}+u_{3}-u_{4})]^\mathrm{T}\nonumber
\end{align}
where $\boldsymbol{u}=\left[u_{1}, u_{2}, u_{3}, u_{4}\right]^\mathrm{T}$ is the control input, $k_t$, $k_m$ and $l$ are model parameters which are positive real numbers. In particular, $u_i=\Omega_i^2$ for all $i\in[1,4]_d$, where $\Omega_i$ are the angular speeds of the motors. 

Now consider the modified control input $\widetilde{\boldsymbol{u}}=[\widetilde{u}_1,\;\widetilde{u}_2,\;\widetilde{u}_3,\;\widetilde{u}_4]^\mathrm{T}$ defined as follows:
\begin{align}
&\boldsymbol{\dbtilde{u}}=[\widetilde{u}_1,\;\widetilde{u}_2,\;\widetilde{u}_3]^\mathrm{T}=M+J\omega^\times\boldsymbol{\omega},\nonumber\\
&  \widetilde{{u}}_4={(\dot{\boldsymbol{v}}})_3=\boldsymbol{e}_3^\mathrm{T}\big( (1/m) F \boldsymbol{e}_{3}-\omega^\times \boldsymbol{v}-\boldsymbol{g} R^\mathrm{T} \boldsymbol{e}_{3}\big).
\label{eqn:u_bar_4}
\end{align}
Then, the quadrotor dynamics given in Eqn. \eqref{eqn:nonlinear_quad_dynamics} can be written more compactly as follows:
\begin{align}
& \dot{h}=hS,~~~~
[\dot{\boldsymbol{\omega}}^\mathrm{T}\;\;{(\boldsymbol{\dot{v}})_3}^\mathrm{T}]^\mathrm{T}=\mathcal{J}\widetilde{\boldsymbol{u}},
\label{eqn:quadrotor_dynamics_compact}
\end{align}
where $\mathcal{J}\in\mathbb{R}^{4\times4}$, $S\in\mathbb{R}^{4\times4}$ and $h\in SE(3)$ are given by
\begin{align*}
  \mathcal{J}:=  \left[
\begin{smallmatrix}
J^{-1} & 0\\
\mathbf{0}& 1
\end{smallmatrix}\right],\;\;     S:=\left[\begin{smallmatrix}
{\omega}^\times & \boldsymbol{v} \\
\mathbf{0} & 0
\end{smallmatrix}\right],\;\;   h:=\left[\begin{smallmatrix}
R & \boldsymbol{p} \\
\mathbf{0} & 1
\end{smallmatrix}\right].
\end{align*}

\subsection{Brief review of Koopman operator\label{sec:koopman_operator_review}}
{Consider a controlled dynamical system $\dot{\boldsymbol{x}}=\boldsymbol{f}(\boldsymbol{x}, \boldsymbol{u})$ whose state evolution is described as follows:}
\begin{align}
\boldsymbol{x}\left(t_{i}+t_{s}\right) &= F_{t_s}\left(\boldsymbol{x}\left(t_{i}\right); \boldsymbol{u}_{\left[t_{i}, t_i+t_s\right]}\right)\nonumber \\
&=\boldsymbol{x}\left(t_{i}\right)+\int_{t_{i}}^{t_{i}+t_{s}} f(\boldsymbol{x}(\sigma), \boldsymbol{u}(\sigma) \mathrm{d} \sigma\nonumber
\end{align}
{where $t_i$
is the $i^{\text{th}}$ sampling time and $t_s \geq 0$ is the sampling interval, $\boldsymbol{x}(t)\in\mathbb{R}^n$ is the state of the system at time $t$,
$\boldsymbol{u}(t)\in\mathbb{R}^m$ is the control input at time $t$ and $\boldsymbol{u}_{\left[t_{i}, t_i+t_s\right]} :=\{\boldsymbol{u}(t): t\in [t_i+t_s] \}$, and finally, $F_{t_s}$
is the flow mapping which advances the state of the $\boldsymbol{x}$ at time $t_i$ to the state $\boldsymbol{x}'$ at time $t = t_i + t_s$ when the system is driven by the input signal $\boldsymbol{u}(t)$ for $t \in [t_i,t_i+t_s]$. The Koopman operator $\mathcal{K}$ is a linear infinite-dimensional (composition) operator that acts on functions known as observables or observable functions which belong to a function space $\mathcal{F}$. We refer to $\mathcal{F}$ as the lifted space or the space of observables. In particular, given a countably infinte collection of observable functions $\boldsymbol{b}=(b_1(\boldsymbol{x}(t)), b_2(\boldsymbol{x}(t)), \dots)$, where $b_{i}:\mathbb{R}^n\rightarrow\mathbb{R}$, the Koopman operator $\mathcal{K}: \mathcal{F} \rightarrow \mathcal{F}$ is defined as follows:}
\begin{align}
    \mathcal{K} \boldsymbol{b}\left(\boldsymbol{x}\left(t_{i}\right)\right)=\boldsymbol{b}\left(F_{t_s}\left(\boldsymbol{x}\left(t_{i}\right); \boldsymbol{u}_{\left[t_{i}, t_i+t_s\right]}\right)\right)=\boldsymbol{b}\left(\boldsymbol{x}\left(t_{i+1}\right)\right). \nonumber
\end{align}
{Note that the Koopman operator $\mathcal{K}$ is a linear operator which takes the observation of state ${b}_i(\boldsymbol{x}(t_i))$ at time $t=t_i$ and shifts it to the next observation at time $t=t_{i+1}$. In contrast to standard linearization methods used for approximation of nonlinear systems, which become inaccurate away from the linearization point, the Koopman operator can describe the exact evolution of the observables of a nonlinear system globally (except from, perhaps, chaotic systems). Finding, however, a suitable set of observables that can span the lifted space $\mathcal{F}$ can be a challenging task as there is no systematic way to construct these functions in general. In practice, one has to work with a finite collection of observables (truncation of the countably infinite collection of observables that span $\mathcal{F}$), which is represented as a vector  $\boldsymbol{b}(\boldsymbol{x})\in\mathbb{R}^N$, where  
$\boldsymbol{b}(\boldsymbol{x})=[{b}_1(\boldsymbol{x}),\;{b}_2(\boldsymbol{x}), \dots, {b}_{N}(\boldsymbol{x})]^\mathrm{T},\nonumber$ for $N\gg n$. The vector $\boldsymbol{b}(x)$ is often referred to as the \textit{lifted state} as it corresponds to the state of the system in the lifted state space.}

\section{Derivation of observable functions\label{sec:derivation_of_observable_functions}}
 We will now present a systematic way to derive a sequence of observable functions that will allow us to form the lifted space for the quadrotor dynamics governed by Eqn. \eqref{eqn:quadrotor_dynamics_compact}.
 
 \begin{theorem}
 {The lifted (function) space of the quadrotor dynamics is spanned by the following (countably infinite) collection of observable (basis) functions:}
 \begin{align}
     {\boldsymbol{b}= (b_1, b_2, \dots) =\left( \boldsymbol{\omega},(\boldsymbol{v})_3, \{\underline{g}_{k}\}_{k=0}^{\infty}, \{ f_k \}_{k=0}^{\infty} \right),}
     \label{eq:basisfun}
 \end{align}
 {where $g_k=hS^k$ and $f_k=\omega^{\times k}\boldsymbol{v}$, for $k\in \mathbb{N}$.}
 \end{theorem}
\begin{proof}

Let $g_0:=h$. Then, $\dot{g_0}=\dot{h}=hS$. Let $g_1:= hS$.
Then,
\begin{align}
\dot{g}_{1} &=\dot{h}S+h\dot{S}=g_{1} S+h\dot{S}=g_2+h\dot{S},\nonumber
\end{align}
where ${g}_{2}:=g_{1} S $. Therefore, the time derivative of $g_{k}:=g_{k-1}S$ for all $k\in[1,N_1]_d$ satisfies
\begin{align}
\dot{g}_{k}= g_{k+1}+h \sum_{i=1}^{k} S^{(i-1)}\dot{S} S^{(k-i)}.
\label{eqn:g_k_attitude}
\end{align}
Now, vectorization of the matrices $h$ and $g_{k}$ yields
\begin{align}
\underline{\dot{{h}}}&=\underline{g}_{1},\quad \dot{\underline{g}}_{k}=\underline{g}_{k+1}+B_{k} \widetilde{\boldsymbol{u}},\quad\forall k\in[1,N_1]_d
\label{eqn:g_k_attitude_underline}
\end{align}
where $B_k\widetilde{\boldsymbol{u}}$ is obtained by vectorization of the matrix $h \sum_{i=1}^{k} S^{(i-1)}\dot{S} S^{(k-i)}$. 
Now let $f_0:=\boldsymbol{v}$. Then,
\begin{align}
    \dot{f_0}=\dot{\boldsymbol{v}}=-\omega^\times\boldsymbol{v}=-f_1,\nonumber
\end{align}
where $f_1:=\omega^\times f_0$, which implies that
\begin{align}
    \dot{f}_1&=\dot{{\omega}}^\times\boldsymbol{v}+{{\omega}}^\times\dot{\boldsymbol{v}}=({J}\dbtilde{\boldsymbol{u}})^\times\boldsymbol{v}-(\omega^\times)^2\boldsymbol{v}.\nonumber
\end{align}

Let $f_k:={\omega}^{\times } f_{k-1}$. Then, it follows that 
\begin{align}
    \dot{f_k}=-f_{k+1}+\sum_{i=1}^{k}{\omega}^{\times {(i-1)}}({J}\dbtilde{\boldsymbol{u}})^\times {\omega}^{\times {(k-i)}}\boldsymbol{v},
    \label{eqn:f_k_position}
\end{align}
for $k\in[1,N_2]_d$.
{By taking $N_1\rightarrow\infty$ and $N_2\rightarrow\infty$, we obtain the countably infinite collection of observable functions defined in \eqref{eq:basisfun} which spans the lifted (function) space of the quadrotor dynamics on $SE(3)$.}

\end{proof}
\begin{remark}
{Note that the observables $f_k$ are associated with the linear and angular velocities of the quadrotor whereas the observables $g_k$ with its position and attitude. Now, the finite (truncated) collection of observable / basis functions can be written as follows: $\boldsymbol{b}=$ $\left( b_1,\dots, b_N\right)=\big( \boldsymbol{\omega},(\boldsymbol{v})_3, \{\underline{g}_{k}\}_{k=0}^{N_1-1}, \{ f_k \}_{k=0}^{N_2-1} \big)$, where $[b_1,b_2,b_3]^\mathrm{T}:=\boldsymbol{\omega}$, $b_4:=(\boldsymbol{v})_3$, $[b_5\dots b_{20}]^\mathrm{T}:=\underline{g}_{0}, \dots$, ${[b_{16N_1+4+3(N_2-1)},\dots,b_N]^\mathrm{T}=\underline{f}_{N_2-1}}$ with $N=16N_1+3N_2+4$.}
\end{remark}
Based on Eqns. \eqref{eqn:g_k_attitude} and \eqref{eqn:f_k_position}, the lifted-state space linear dynamics for the quadrotor is given by
\begin{align}
    \dot{\boldsymbol{\mathcal{X}}}=\mathcal{A}\boldsymbol{\mathcal{X}}+\mathcal{B}\boldsymbol{\widetilde{u}},
    \label{eqn:lifted_space_dynamics}
\end{align}
where $\boldsymbol{\mathcal{X}}\in\mathbb{R}^N$ is the lifted state for the quadrotor dynamics and $\mathcal{A}\in\mathbb{R}^{N\times N}$ and $\mathcal{B}\in\mathbb{R}^{N\times 4}$ are given by
\begin{align}
& \boldsymbol{\mathcal{X}}=[ \boldsymbol{\omega}^\mathrm{T},(\boldsymbol{v})_3,\underline{g}_{0}^\mathrm{T}, \dots, \underline{g}_{{N_1-1}}^{\mathrm{T}},{f_0}^\mathrm{T},\dots,{f_{N_2-1}}^\mathrm{T} ]^\mathrm{T},\label{eqn:observables} \\
& \mathcal{A}=\mathrm{bdiag}(A_1,A_2,A_3),\;\; \mathcal{B}=[B,\;\mathbf{0}_{16\times 4},\;B_1,\;\dots,B_{N'}]^\mathrm{T},\nonumber
\end{align}
where $N={4+16N_1+3N_2}$, $N'=N_1+N_2$ and
\begin{align}
    & A_1=\mathcal{J},\quad A_2=\left[\begin{smallmatrix}\mathbf{0}_{16(N_1-1)\times16} & \mathbf{I}_{16(N_1-1)}\\
    \mathbf{0}_{16} & \mathbf{0}_{16\times 16(N_1-1)}\end{smallmatrix}\right],\nonumber\\
    & A_3=\left[\begin{smallmatrix}\mathbf{0}_{3(N_2-1)\times3} & -\mathbf{I}_{3(N_2-1)}\\
    \mathbf{0}_{3} & \mathbf{0}_{3\times 3(N_2-1)}\end{smallmatrix}\right].\nonumber
\end{align}
{Note that \eqref{eqn:lifted_space_dynamics} is an approximation of the original dynamics of the quadrotor given in \eqref{eqn:nonlinear_quad_dynamics} }. Furthermore, the matrices $B_k$ are such that $B_k\widetilde{\boldsymbol{{u}}}$ satisfies:
\begin{align}
    B_k\widetilde{\boldsymbol{{u}}}=\left\{ \begin{array}{l}
\texttt{vec}(h \sum_{i=1}^{k} S^{(i-1)}\dot{S} S^{(k-i)}),\;\; k\in[1,N_1]_d
\\
  \sum\limits_{i=1}^{k-N_1}{\omega}^{\times {(i-1)}}({J}\dbtilde{\boldsymbol{u}})^\times {\omega}^{\times {(k-i)}}\boldsymbol{v},
    k\in[N_1+1,N']_d
\end{array}
\right.\nonumber
\end{align}
The following result will be useful in subsequent discussion.
\begin{prop}
For all $k\in[1,N']_d$, the matrix $B_k$ is only state dependent. Thereafter, $\mathcal{B}$ is also a state-dependent matrix only.
\label{prop:B_is_state_dependent}
\end{prop}
\begin{proof}
We have
\begin{align}
{\omega}^{\times {(i-1)}}({J}\dbtilde{\boldsymbol{u}})^\times {\omega}^{\times {\ell}}\boldsymbol{v}=(-1)^{\ell+1}{\omega}^{\times {(k-i)}}v^\times{\omega}^{\times {(i-1)}}{J}\dbtilde{\boldsymbol{u}}, \nonumber
\end{align}
where $\ell:=k-i$. Therefore,
\begin{align}
    B_k=\mathrm{bdiag}\big(\sum_{i=1}^{k-N_1}\small(-1\small)^{\ell-N_1+1}{\omega}^{\times {(\ell-N_1)}}v^\times{\omega}^{\times {(i-1)}}{J},0\big),\nonumber
\end{align}
for all $k\in[N_1+1,N']_d$. In addition,
\begin{align}
   & h S^{(i-1)}\dot{S} S^{(k-i)}\nonumber\\
     & =(-1)^{\ell}\left[\begin{array}{cc}R\omega^{\times(k-1)}{J}\dbtilde{{u}}^{\times} & -R\omega^{\times(\ell-1)}v^\times\omega^{\times(i-2)}{J}\dbtilde{\boldsymbol{u}}\\
    \mathbf{0} & 0\end{array}\right]\nonumber
\end{align}
 Let $C_{1,ik}=(-1)^{\ell}R\omega^{\times(k-1)}{J}$ and $C_{2,ik}=(-1)^{(\ell-1)}R\omega^{\times(\ell-1)}v^\times\omega^{\times(i-2)}{J}$.
Then,
\begin{align*}
 \texttt{vec}((-1)^{\ell}R\omega^{\times(k-1)}{J}\dbtilde{{u}}^{\times})=(\textbf{I}_3\otimes C_{1,ik})C_3\boldsymbol{\widetilde{u}},  
\end{align*}
where $C_3\in\mathbb{R}^{9\times 4}$ is a constant matrix such that $\texttt{vec}(\dbtilde{{u}}^\times)=C_3\widetilde{\boldsymbol{u}}$.
Therefore, $B_k\in\mathbb{R}^{16\times 4}$ is given by
\begin{align}
    B_k=\sum_{i=1}^{k}[(\textbf{I}_3\otimes C_{1,ik})C_3,\;\;[C_{2,ik}, \mathbf{0}_{3\times 1}]^\mathrm{T},\;\;\mathbf{0}_{4}]^\mathrm{T},\nonumber
\end{align}
for $ k\in[1,N_1]_d$. This completes the proof.
\end{proof}
\begin{remark}
Since $\mathcal{B}$ is a state-dependent matrix (from Proposition \ref{prop:B_is_state_dependent}), let us consider the following input transformation $\boldsymbol{U}^\star=\mathcal{B}\boldsymbol{\widetilde{u}}$. Then, Eqn. \eqref{eqn:lifted_space_dynamics} can be written as follows:
\begin{align}
    \dot{\boldsymbol{\mathcal{X}}}=\mathcal{A}\boldsymbol{\mathcal{X}}+\widetilde{\mathcal{B}}\boldsymbol{U}^\star,
    \label{eqn:lifted_space_dyn_final}
\end{align}
where $\widetilde{\mathcal{B}}= \mathrm{bdiag}(\mathbf{I}_4,\mathbf{0}_{16},\mathbf{I}_{16(N_1-1)},\mathbf{0}_{3},\mathbf{I}_{3(N_2-1)})$.
To realize the control input $\boldsymbol{\mathcal{U}}$ from $\boldsymbol{U}^\star$, one can solve the following least-squares optimization problem
\begin{align}\label{prob:least}
    \text{minimize:}\;\;(\mathcal{B} \boldsymbol{\widetilde{u}}-\widetilde{{\mathcal{B}}} \boldsymbol{U}^\star)^{\mathrm{T}}(\mathcal{B} \boldsymbol{\widetilde{u}}-\widetilde{{B}} \boldsymbol{U}^\star), 
\end{align}
whose solution is given by $
\boldsymbol{\widetilde{u}}=\mathcal{B}^{\dagger} \widetilde{\mathcal{B}} \boldsymbol{U}^\star$.
\end{remark}

\subsection{Point-wise Convergence}
To truncate the lifted space dynamics, we must first ensure that the terms which will be truncated are close to zero.
Let us define the following sets $\mathcal{D}_{\omega}$, and $\mathcal{D}_{v}$ as follows:
\begin{align}
    \mathcal{D}_{\omega}&=\{\boldsymbol{\omega}\in\mathbb{R}^3:|\boldsymbol{\omega}|\leq\Bar{\omega}\},\;\;\; 
    \mathcal{D}_{v}=\{\boldsymbol{v}\in\mathbb{R}^3:|\boldsymbol{v}|\leq\Bar{v}\},\nonumber
\end{align}
where $\Bar{\omega}<1/\sqrt{2}$ and $\Bar{v}<1$. {In addition, we will assume that $|\dbtilde{\boldsymbol{u}}|$ is upper bounded due to practical actuator constraints.}
\begin{theorem}
\textit{
{Let us assume that there exists $c>0$ such that $|J\dbtilde{\boldsymbol{u}}|\leq c$.} The sequences of functions $\underline{g}_k$, $\dot{\underline{g}}_k$, ${f}_k$ and $\dot{f}_k$ converge pointwise to $\mathbf{0}$ as $k$ tends to infinity, that is,
\begin{align}
     & \underset{k\to\infty}{\lim}\underline{g}_k=\mathbf{0},\;\; \underset{k\to\infty}{\lim}\dot{\underline{g}}_k=\mathbf{0}\;\;\forall\; h\in SE(3),\;\; S\in\mathcal{D}_{\omega}\times\mathcal{D}_{v}\nonumber\\
     & \underset{k\to\infty}{\lim}{f}_k=\mathbf{0},\quad\underset{k\to\infty}{\lim}\dot{f}_k=\mathbf{0}\quad\forall \;\; \boldsymbol{\omega}\in\mathcal{D}_{\omega},\;\; \boldsymbol{v}\in\mathcal{D}_v\nonumber
\end{align}}
\label{prop:f_k_and_g_k_equals_0}
\end{theorem}
\begin{proof}
Since ${\omega}^{\times k}\boldsymbol{v}= \underbrace{\boldsymbol{\omega}\times\ldots \times\boldsymbol{\omega}}_{k \text{ times }}\times\boldsymbol{v}$, we have
\begin{align}
|{\omega}^{\times k}\boldsymbol{v}| \leq |\boldsymbol{\omega}|^k|\boldsymbol{v}|.
    \label{eqn:omega_k_less_than}
\end{align}
Because $|\boldsymbol{\omega}|<1/\sqrt{2}$, $|\boldsymbol{v}|<1$ and $ f_k={\omega}^{\times k}\boldsymbol{v}$, we can conclude that $\underset{k\to\infty}{\lim}|f_k|=0$, which implies $\underset{k\to\infty}{\lim}f_k=\mathbf{0}.$
From Eqn. \eqref{eqn:f_k_position}, $\dot{f}_k$ can be written as follows:
\begin{align}
{\dot{f}}_k=-{f}_{k+1}+\sum_{i=1}^{k}{\omega}^{\times {(i-1)}}({J}\dbtilde{\boldsymbol{u}})^\times {\omega}^{\times {(k-i)}}\boldsymbol{v}.
\label{eqn:f_k_limit}
\end{align}
Now,
\begin{align}
    |{\omega}^{\times {(i-1)}}({J}\dbtilde{\boldsymbol{u}})^\times {\omega}^{\times {(k-i)}}\boldsymbol{v}|\leq |\boldsymbol{\omega}|^{(k-1)}|{J}\dbtilde{\boldsymbol{u}}||\boldsymbol{v}|.
    \label{eqn:omega_x_j_omega_x_bound}
\end{align}
 In view of \eqref{eqn:f_k_limit}, \eqref{eqn:omega_x_j_omega_x_bound} implies
\begin{align}
|\sum_{i=1}^{k}{\omega}^{\times {(i-1)}}({J}\dbtilde{\boldsymbol{u}})^\times {\omega}^{\times {(k-i)}}\boldsymbol{v}|\leq
k |\boldsymbol{\omega}|^{(k-1)}{c}|\boldsymbol{v}|.
\label{eqn:vec_less_than_k_omega}
\end{align}
Because $|\boldsymbol{\omega}|<1/\sqrt{2}$ and $\underset{k\to\infty}{\lim}(k|\boldsymbol{\omega}|^k)=0$ for all $\boldsymbol{\omega}\in\mathcal{D}_{\omega}$, taking limits on both sides of \eqref{eqn:vec_less_than_k_omega} gives
\begin{align}
    \underset{k\to\infty}{\lim}\big|\sum_{i=1}^{k}{\omega}^{\times {(i-1)}}({J}\dbtilde{\boldsymbol{u}})^\times {\omega}^{\times {(k-i)}}\boldsymbol{v}\big|=0.
    \label{eqn:summation_omega_times_limit}
\end{align}
In view of Eqn. \eqref{eqn:f_k_position}, we have
\begin{align}
    |{\dot{f}}_k|\leq|{f}_{k+1}|+\big|\sum_{i=1}^{k}{\omega}^{\times {(i-1)}}({J}\dbtilde{\boldsymbol{u}})^\times{\omega}^{\times {(k-i)}}\boldsymbol{v}\big|.
    \label{eqn:f_k_lessthan_f_k_plus_1}
\end{align}
Using Eqns. \eqref{eqn:f_k_limit}, \eqref{eqn:summation_omega_times_limit} and taking limits on both sides of $\eqref{eqn:f_k_lessthan_f_k_plus_1}$ gives $\underset{k\to\infty}{\lim} |{\dot{f}}_k|\leq0.$
Therefore, $\underset{k\to\infty}{\lim} |\dot{f}_k|=0$, which implies $\underset{k\to\infty}{\lim} \dot{f}_k=\mathbf{0}$.

The expression for $g_k$ can be written as follows:
\begin{align}
    g_k=\left[\begin{smallmatrix}
          R{\omega}^{\times k} & R{\omega}^{\times (k-1)}\boldsymbol{v} \\
          \mathbf{0} & 0
    \end{smallmatrix}\right]=\left[\begin{smallmatrix}
          R{\omega}^{\times k} & Rf_{k-1} \\
          \mathbf{0} & 0
    \end{smallmatrix}\right].
    \label{eqn:g_k_matrix_expression}
\end{align}
Since the rotation matrix $R$ is orthogonal, we have
\begin{align}
    \|R\omega^{\times k}\|_F&=|\texttt{vec}(R\omega^{\times k})|=\sqrt{\texttt{tr}((R\omega^{\times k})^\mathrm{T} R\omega^{\times k})}\nonumber\\
    &=\sqrt{\texttt{tr}(\omega^{\times k\mathrm{T}}R^\mathrm{T} R\omega^{\times k})}=\|\omega^{\times k}\|_F\nonumber\\
    &=\|\omega^{\times k}\|_F=|\texttt{vec}(\omega^{\times k})|.
    \label{eqn:R_omega_equals_omega}
\end{align}
Because the Frobenius norm is a submultiplicative norm, 
\begin{align}
\|\omega^{\times k}\|_F\leq\|\omega^\times\|_F^k=|\texttt{vec}(\omega^\times)|^k=\sqrt{2}^k|\boldsymbol{\omega}|^k .
\label{eqn:omega_bound}
\end{align}
Using \eqref{eqn:R_omega_equals_omega} and \eqref{eqn:omega_bound}, it follows that
\begin{align}
   |\texttt{vec}(R\omega^{\times k})|\leq|\texttt{vec}(\omega^\times)|^k= \sqrt{2}^k |\boldsymbol{\omega}|^k.
   \label{eqn:R_omegax_bound}
\end{align}
Since $|\boldsymbol{\omega}|<1/\sqrt{2}$, taking limits on both sides of \eqref{eqn:R_omegax_bound} gives
\begin{align}
    & \underset{k\to\infty}{\lim}|\texttt{vec}(R{\omega}^{\times k})|=\underset{k\to\infty}{\lim}|\texttt{vec}({\omega}^{\times k})|=0.
    \label{eqn:limit_r_omega_x}
\end{align}
From \eqref{eqn:omega_k_less_than} we have, $|Rf_{k-1}|=|f_{k-1}|\leq|\boldsymbol{\omega}|^{k-1}|\boldsymbol{v}|.$
Therefore, $\underset{k\to\infty}{\lim}|\texttt{vec}(R{f}_{k-1})|=0$ and we can conclude that $\underset{k\to\infty}{\lim}|\underline{S}^k|=0$ (using \eqref{eqn:g_k_matrix_expression} and \eqref{eqn:limit_r_omega_x}) which implies
    $\underset{k\to\infty}{\lim}|\underline{g}_k|=0$ and thus $ \underset{k\to\infty}{\lim}\underline{g}_k=\mathbf{0}$.
Now consider ${\dot{\underline{g}}_k}$ (from Eqn. \eqref{eqn:g_k_attitude_underline}) which is given as follows:
\begin{align}
    \underline{\dot{g}}_{k}= \underline{g}_{k+1}+ \texttt{vec}\big(h\sum_{i=1}^{k} S^{(i-1)}\dot{S} S^{(k-i)}\big).
    \label{eqn:g_k_dot_equation}
\end{align}
Now since 
$\underset{k\to\infty}{\lim}|\underline{S}^k|=0$, it follows that
\begin{align}
   \|S^{(i-1)}\dot{S} S^{(k-i)}\|_F=|\texttt{vec}(S^{(i-1)}\dot{S} S^{(k-i)})|\leq |\underline{\dot{S}}| |\underline{S}^{(k-1)}|.\nonumber
\end{align}
Taking limits on both sides, we have
\begin{align}
    \underset{k\to\infty}{\lim}|\texttt{vec}(\sum_{i=1}^{k} S^{(i-1)}\dot{S} S^{(k-i)})|\leq  \underset{k\to\infty}{\lim} k |\underline{\dot{S}}| |\underline{S}^{(k-1)}|=0.\nonumber
\end{align}
Hence, from Eqn. \eqref{eqn:g_k_dot_equation} we get
\begin{align}
   \underset{k\to\infty}{\lim}|\dot{\underline{g}}_k|&\leq  \underset{k\to\infty}{\lim}|\underline{g}_{k+1}|+ \underset{k\to\infty}{\lim}|\texttt{vec}(\sum_{i=1}^{k} S^{(i-1)}\dot{S} S^{(k-i)})| \nonumber\\
   & \leq0.\nonumber
\end{align}
Thus, $\underset{k\to\infty}{\lim}|\dot{\underline{g}}_k|=0$ which implies $\underset{k\to\infty}{\lim}\dot{\underline{g}}_k=\mathbf{0}$ and the theorem is proved.
\end{proof}

\begin{remark}
Theorem \ref{prop:f_k_and_g_k_hat} allows us to truncate the proposed sequence of observables that span the lifted space to obtain a lower (finite) dimensional linear state space model with lifted state $\boldsymbol{\omega}\in\mathcal{D}_{{\omega}}$ and $\boldsymbol{v}\in\mathcal{D}_{{v}}$ for higher values of $N$. However this is not applicable for all $\boldsymbol{\omega}$ and $\boldsymbol{v}$.
\end{remark}

\subsection{Point-wise convergence for constrained case}
For most practical applications, the magnitude of the angular and linear velocities are constrained due to actuation limitations. In other words, there exists some constants ${\omega}_0$ and $v_0$ such that
\begin{align}
 {\omega}_0>\sqrt{2}\underset{\boldsymbol{\omega}}{\texttt{max}}(|\boldsymbol{\omega}|),\;\quad  {v}_0>\underset{\boldsymbol{v}}{\texttt{max}}(|\boldsymbol{v}|). \nonumber
\end{align}
Therefore, for higher $|\boldsymbol{\omega}|$ and $|\boldsymbol{v}|$, these two terms can be normalized so that the truncation can be made possible as illustrated later. We now define $\widehat{\boldsymbol{\omega}}:=\boldsymbol{\omega}/\omega_0$ and $\widehat{\boldsymbol{v}}:=\boldsymbol{v}/v_0$. New observables $\widehat{g}_k$ and $\widehat{f}_k$ are then defined as follows:
\begin{align}
   & \widehat{g}_k= hS^k/s_0^k=h\widehat{S}^k,\;\;\widehat{f}_k={\omega}^{\times k}\boldsymbol{v}/{\omega}_0^k v_0=\widehat{{\omega}}^{\times k}\widehat{\boldsymbol{v}},
   \label{eqn:f_k_hat}
\end{align}
where $s_0=\texttt{max}\{{\omega}_0,v_0\}$. Next, we prove that both $\widehat{g}_k$ and $\widehat{f}_k$ tend to zero for all $(\boldsymbol{\omega},\boldsymbol{v})\in\mathcal{D}_{\widehat{\omega}}\times \mathcal{D}_{\widehat{v}}$ as $k\rightarrow\infty$ where
\begin{align}
    \mathcal{D}_{\widehat{\omega}}&=\{\boldsymbol{\widehat{\omega}}\in\mathbb{R}^3:|\boldsymbol{\widehat{\omega}}|\leq\Bar{{\omega}}\},\;\;
    \mathcal{D}_{\widehat{v}}=\{\boldsymbol{\widehat{v}}\in\mathbb{R}^3:|\boldsymbol{\widehat{v}}|\leq\Bar{{v}}\},\nonumber
\end{align}
where $\Bar{\omega}<1/\sqrt{2}$ and $\Bar{v}<1$.

\begin{theorem}
\textit{For any $\boldsymbol{\omega}\in\mathcal{D}_{\widehat{\omega}}$ and $\boldsymbol{v}\in\mathcal{D}_{\widehat{v}}$, the sequences of functions $\widehat{\underline{g}}_k$, $\dot{\widehat{\underline{g}}}_k$, $\widehat{f}_k$ and $\dot{\widehat{f}}_k$ converge pointwise to $\mathbf{0}$, i.e,
\begin{align}
    &\underset{k\to\infty}{\lim}\underline{\widehat{g}}_k=\mathbf{0},\;\; \underset{k\to\infty}{\lim}\underline{\dot{\widehat{g}}}_k=\mathbf{0}\;\;\forall\;h\in SE(3), S\in\mathcal{D}_{\widehat{\omega}}\times \mathcal{D}_{\widehat{v}}\nonumber \\
   &\underset{k\to\infty}{\lim}{\widehat{f}}_k=\mathbf{0},\quad  \underset{k\to\infty}{\lim}{\dot{\widehat{f}}}_k=\mathbf{0}\quad\forall\;\boldsymbol{\omega}\in\mathcal{D}_{\widehat{\omega}},\;\; \boldsymbol{v}\in\mathcal{D}_{\widehat{v}}\nonumber
\end{align}}
\label{prop:f_k_and_g_k_hat}
\end{theorem}
\begin{proof}
Since $|\widehat{\boldsymbol\omega}|<1/\sqrt{2}$ and $|\widehat{\boldsymbol{v}}|<1$, using Theorem \ref{prop:f_k_and_g_k_equals_0}, we conclude that
    $\underset{k\to\infty}{\lim}|\underline{\widehat{g}}_k|=0$, and $\underset{k\to\infty}{\lim}|{\widehat{f}}_k|=0$.
Now $\dot{\widehat{{g}}}_k$ and $\dot{\widehat{f}}_k$ can be written as
\begin{align}
     & \dot{\widehat{f}}_k=-\omega_0\widehat{{f}}_{k+1}+\sum_{i=1}^{k}{\widehat{{\omega}}^{\times {(i-1)}}({J}\dbtilde{\boldsymbol{u}})^\times \widehat{{\omega}}^{\times {(k-i)}} }\widehat{\boldsymbol{v}},\nonumber\\
   & \dot{\widehat{{g}}}_k= s_0\widehat{g}_{k+1}+h \sum_{i=1}^{k} \widehat{S}^{(i-1)}\widehat{\dot{S}} \widehat{S}^{(k-i)} .\nonumber
\end{align}
Since $|\widehat{\boldsymbol{\omega}}|<1/\sqrt{2}$ and $|\widehat{\boldsymbol{v}}|<1$, using Theorem \ref{prop:f_k_and_g_k_equals_0}, we can easily conclude that $ \underset{k\to\infty}{\lim}|\underline{\dot{\widehat{g}}}_k|=0$ and $  \underset{k\to\infty}{\lim}|{\dot{\widehat{f}}}_k|=0$.
Hence the proof is complete.
\end{proof}
Therefore, we obtain the following new lifted state which will be used instead of that given in Eqn. \eqref{eqn:observables}:
\begin{align}
   & \boldsymbol{\mathcal{X}}=[ \boldsymbol{\omega}^\mathrm{T},(\boldsymbol{v})_3,\widehat{\underline{g}}_{0}^\mathrm{T}, \cdots, \widehat{\underline{g}}_{{N_1-1}}^{\mathrm{T}},\widehat{{f}}_0^\mathrm{T},\dots,\widehat{{f}}_{N_2-1}^\mathrm{T} ]^\mathrm{T}
\end{align}
and $\mathcal{A}=\mathrm{bdiag}(\widehat{A}_1,\widehat{A}_2,\widehat{A}_3)$, $\mathcal{B}=[B,\mathbf{0}_{16\times 4},\widehat{B}_1,\;\dots,\widehat{B}_{N'}]^\mathrm{T}$,
where 
\begin{align}
    & \widehat{A}_1=A_1,\quad \widehat{A}_2=\omega_0A_2,\quad \widehat{A}_3=s_0A_3,\nonumber\\
    & \widehat{B}_k\widetilde{\boldsymbol{{u}}}=\left\{ \begin{array}{l}
\texttt{vec}(h \sum_{i=1}^{k} \widehat{S}^{(i-1)}\dot{\widehat{S}} \widehat{S}^{(k-i)}),\;\; k\in[1,N_1]_d
\\
  \sum\limits_{i=1}^{k-N_1}{\widehat{\omega}}^{\times {(i-1)}}({J}\dbtilde{\boldsymbol{u}})^\times {\widehat{\omega}}^{\times {(k-i)}}\boldsymbol{\widehat{v}},
    k\in[N_1+1,N']_d
\end{array}
\right.\nonumber
\end{align}
{Following a similar procedure as in Theorem \ref{prop:B_is_state_dependent}, it can be shown that $\widehat{B}_k$ is only state dependent.}
\begin{theorem}
\textit{For any $\boldsymbol{\omega}\in\mathcal{D}_{\widehat{\omega}}$ and $\boldsymbol{v}\in\mathcal{D}_{\widehat{v}}$, the following inequalities holds,}
\begin{align}
  |\widehat{f}_{k+1}|<|\widehat{f}_k|,\quad  |\underline{\widehat{g}}_{k+1}|<|\underline{\widehat{g}}_k|,\;\;\forall\;\; k\in\mathbb{N}.\nonumber
\end{align}
\label{prop:f_k_1_lessthan_f_k}
\end{theorem}
\begin{proof}
Using Eqn. \eqref{eqn:f_k_hat}, we have
\begin{align}
\frac{|\widehat{f}_{k+1}|}{|\widehat{f}_k|}= \frac{|\widehat{\omega}^{\times(k+1)}\widehat{\boldsymbol{v}}|}{|\widehat{\omega}^{\times k}\widehat{\boldsymbol{v}}|}\leq\frac{|\widehat{\boldsymbol{\omega}}||\widehat{\omega}^{\times k}\widehat{\boldsymbol{v}}|}{|\widehat{\omega}^{\times k}\widehat{\boldsymbol{v}}|}=|\widehat{\boldsymbol{\omega}}|<1/\sqrt{2}.\nonumber
\end{align}
Therefore, $ |\widehat{f}_{k+1}|<|\widehat{f}_k|$.
In addition,
\begin{align}
      \widehat{g}_k=h\widehat{S}^k=\left[\begin{smallmatrix}
          R\widehat{\omega}^{\times k} & R\widehat{f}_{k-1} \\
          \mathbf{0} & 0
    \end{smallmatrix}\right].
    \label{eqn:g_k_hat_matrix}
\end{align}
In addition, 
\begin{align}
  \frac{ |\texttt{vec}(R\widehat{\omega}^{\times (k+1)})|}{|\texttt{vec}(R\widehat{\omega}^{\times k})|}&=\frac{|\texttt{vec}(\widehat{\omega}^{\times (k+1)})|}{|\texttt{vec}(\widehat{\omega}^{\times k})|}\leq|\sqrt{2}\widehat{\boldsymbol{\omega}}|<1.\nonumber
\end{align}
Therefore, $|\texttt{vec}(R\widehat{\omega}^{\times (k+1)})|<|\texttt{vec}(R\widehat{\omega}^{\times (k)})|$. Since $ |\widehat{f}_{k}|<|\widehat{f}_{k-1}|$, by using \eqref{eqn:g_k_hat_matrix}, we conclude that  $|\underline{\widehat{g}}_{k+1}|<|\underline{\widehat{g}}_k|$.
Hence the theorem is proved.
\end{proof}
\begin{figure*}[ht]
 \captionsetup[subfigure]{justification=centering}
 \centering
 \begin{subfigure}{0.32\textwidth}
{\includegraphics[scale=0.25]{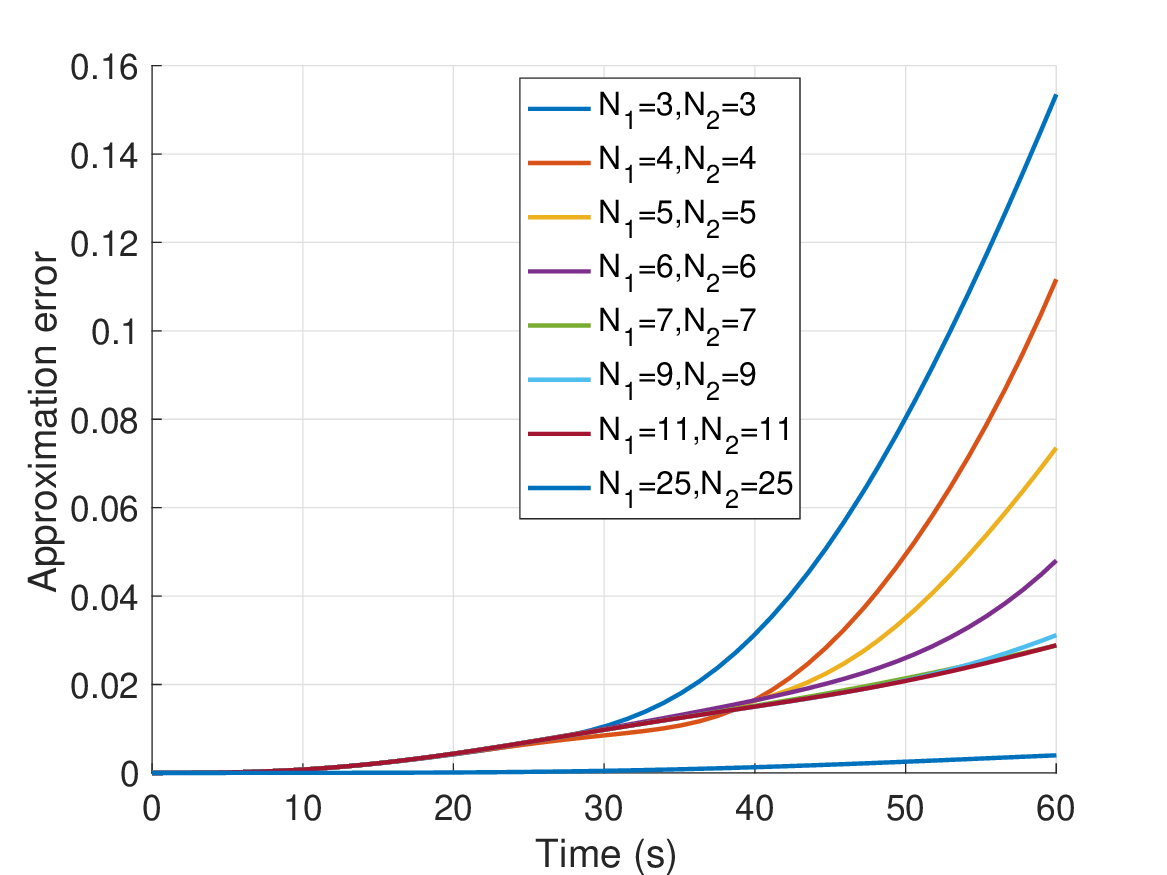}}
\caption{Approximation error of $\boldsymbol{x}$}
\label{fig:approximation_e_x}
 \end{subfigure}
\label{fig:approximation_error_x_position}
 \begin{subfigure}{0.32\textwidth}
{\includegraphics[scale=0.25]{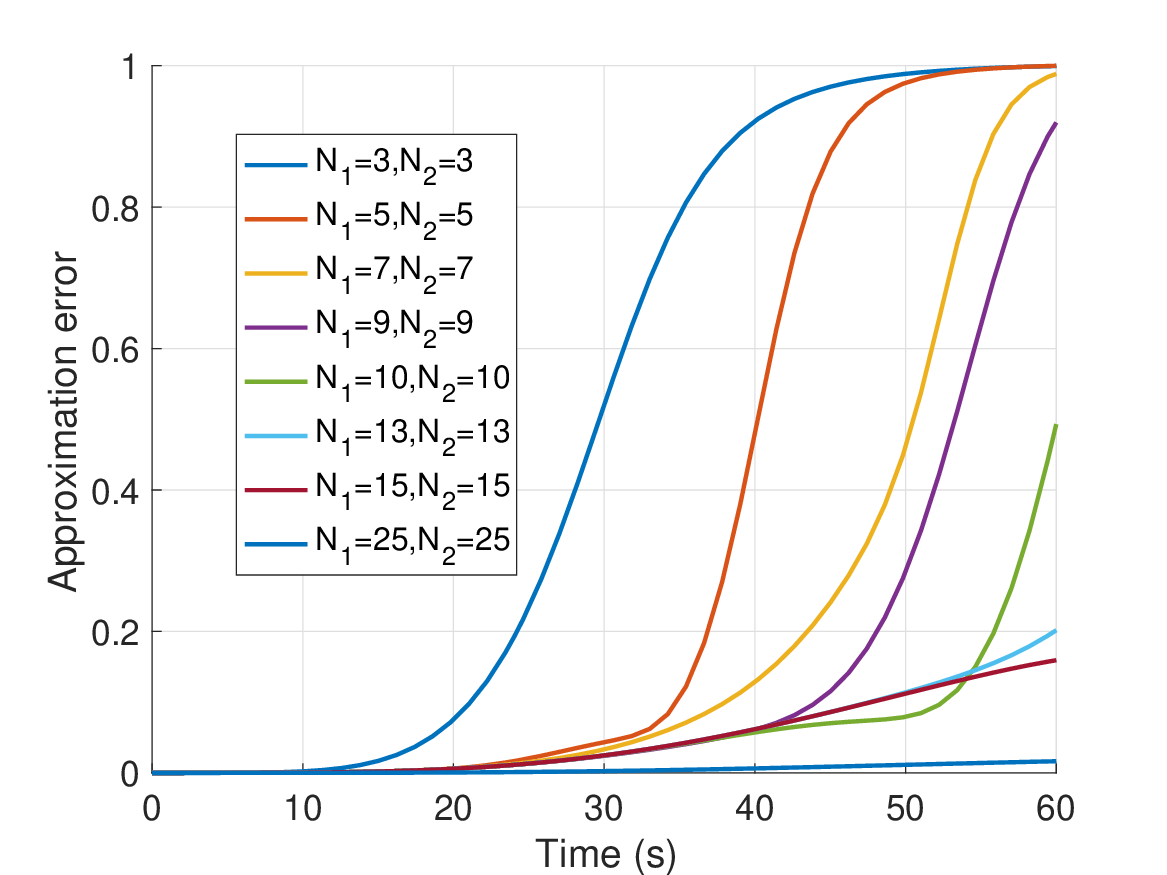}}
 \caption{Approximation error of  $\boldsymbol{v}$}
\label{fig:approximation_e_v}
 \end{subfigure}
\label{fig:approximation_error_v_velocity}
 \begin{subfigure}{0.32\textwidth}
{\includegraphics[scale=0.25]{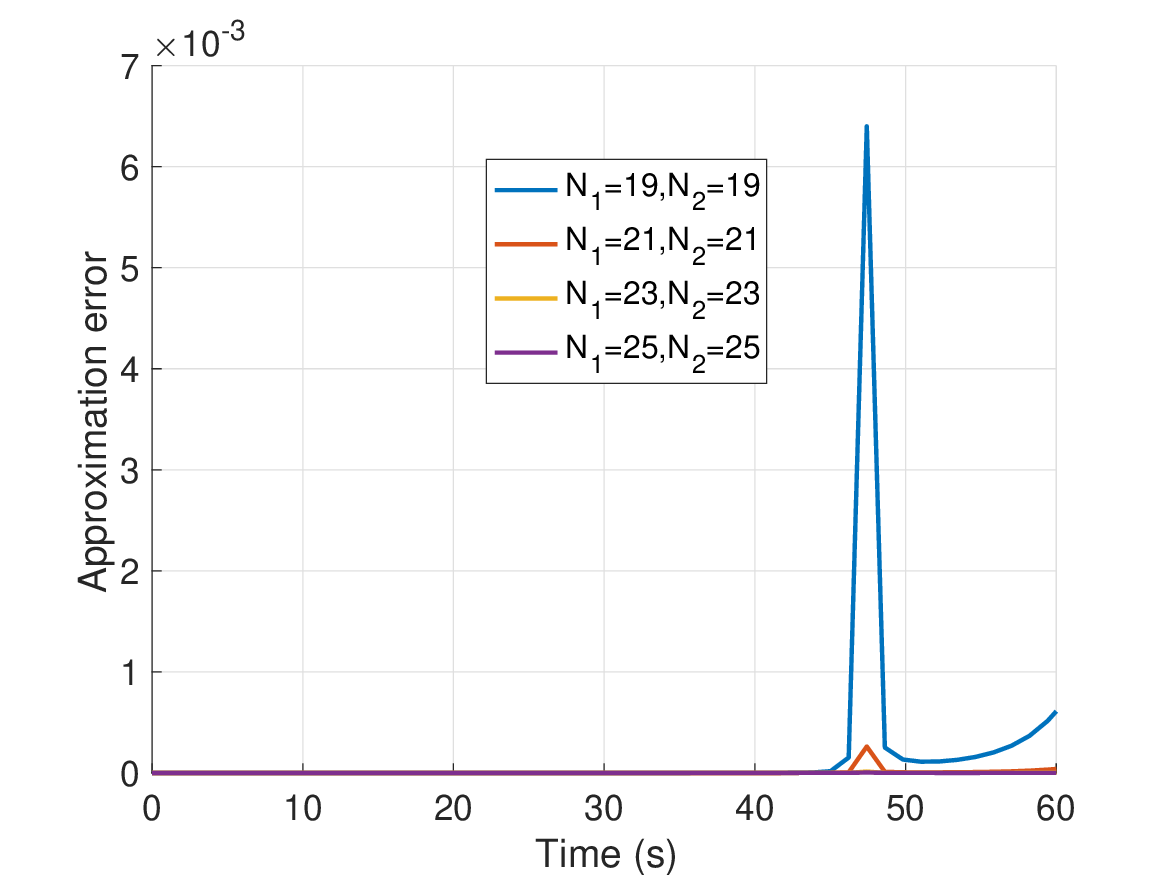}}
\caption{Approximation error of $[\phi,\theta,\psi]^\mathrm{T}$}
\label{fig:approximation_e_angles}
\end{subfigure}
 \caption{Approximation error for position, velocity and Euler angles.}
\label{fig:approximation_error}
\end{figure*}
\begin{remark}
Based on Theorem \ref{prop:f_k_and_g_k_hat} and Theorem \ref{prop:f_k_1_lessthan_f_k}, we can consider truncating the higher dimensional lifted space \eqref{eqn:lifted_space_dyn_final} 
to obtain a lower (finite) dimensional linear state space model for any $\boldsymbol{\omega}\in\mathcal{D}_{\widehat{\omega}}$ and $\boldsymbol{v}\in\mathcal{D}_{\widehat{v}}$. 
\end{remark}

\begin{prop}
\textit{The truncated lifted linear state space model given by Eqn. \eqref{eqn:lifted_space_dyn_final} is controllable for any $\boldsymbol{\omega}\in\mathcal{D}_{\widehat{\omega}}$ and $\boldsymbol{v}\in\mathcal{D}_{\widehat{v}}$. In other words, the pair ($\mathcal{A},\mathcal{\widetilde{B}}$) is controllable.}
\end{prop}
\begin{proof}
The $N\times N^2$ controllability matrix $C$ for the linear model is given as $C=[\widetilde{\mathcal{B}}\;\mathcal{A}\widetilde{\mathcal{B}},\;\mathcal{A}^2\widetilde{\mathcal{B}}\dots\mathcal{A}^{N-1}\widetilde{\mathcal{B}}]$.
The expression for $\mathcal{A}^j\widetilde{\mathcal{B}}$ can be given as follows:
\begin{align}
  \mathcal{A}^j\widetilde{\mathcal{B}}=\left\{ \begin{array}{l}
  \mathrm{bdiag}(\mathbf{0}_{4+16j},\mathbf{I}_{16N_1-16j+3N_2}),\;j\in[1,N'']_d \\
    \mathrm{bdiag}(\mathbf{0}_{4+16N_1+3j},\mathbf{I}_{3N_2-3j}),\;j\in[N''+1,N]_d \\
  \end{array}
  \right.\nonumber
\end{align}
where $N''=N-3N_2$.
The matrices $\mathcal{\widetilde{B}}$ and $\mathcal{A}^j\mathcal{\widetilde{B}}$ have $16N_1-16j+3N_2-3j$ common independent columns and the matrix $\mathcal{A}^j\mathcal{\widetilde{B}}$ has $16N_1-16j+3N_2-3j$ independent columns. In addition, $\mathcal{\widetilde{B}}$ has $4+16+3$ independent columns which are not common to any $\mathcal{A}^j\mathcal{\widetilde{B}}$.
Therefore, rank of $C$ is $4+16N_1+3N_2=N$ and hence the lifted linear state space model is controllable.
This completes the proof.
\end{proof}

\section{Numerical simulations\label{sec:results}}
In this section, we present numerical simulations to validate the proposed Koopman operator based approximation of the nonlinear dynamics \eqref{eqn:nonlinear_quad_dynamics} with the derived linear state space model \eqref{eqn:lifted_space_dyn_final}. Simulation studies have been carried out using MATLAB R2020b on an Intel Core i7 2.2GHz processor. 
 \subsection{Approximation error}
The goodness of fit for the lifted-space dynamics with the exact nonlinear model for the quadrotor was calculated using the approximation error as ${|\boldsymbol{a}-\boldsymbol{b}|}/{|\boldsymbol{b}|}$ 
where $\boldsymbol{a}$ is computed after integrating the lifted-space dynamics \eqref{eqn:lifted_space_dynamics} and $\boldsymbol{b}$ is computed after integrating the nonlinear dynamics \eqref{eqn:nonlinear_quad_dynamics}. A random control input taking values in $[-0.005,0.005]$ is used to propagate the lifted linear state space model and nonlinear dynamics and is given by $\widetilde{\boldsymbol{u}}=0.001\gamma(t)\text{sin}(0.1 t)$ where $\gamma(t)$ is a random number taken from the uniform distribution $[-5,5]$
and the initial conditions are as follows: $R(0)=\mathbf{I}_3,\; \boldsymbol{\omega}(0)=[0.05\;0.05\;0.05]^\mathrm{T},\;\boldsymbol{v}(0)=[0.1\;0.1\;0.1]^\mathrm{T}$. 
Figs. \ref{fig:approximation_e_x}, \ref{fig:approximation_e_v} and \ref{fig:approximation_e_angles} show the variation of approximation error for $\boldsymbol{x}$, $\boldsymbol{v}$ and $[\phi,\theta,\psi]$ respectively for different values of $N_1$ and $N_2$. As expected, we observe that as the dimension of the lifted-space increases, the approximation error decreases.
{We have also used both sinusoidal control inputs and constant signals and observed that the trend in the approximation errors for position, velocity and attitude is nearly the same.}

{Note that irrespective of the dimension $N$, the approximation error becomes larger with time mainly because the integral of the truncated terms increases with time.}

\subsection{Relative error between $\widetilde{\mathcal{B}}\boldsymbol{U}^\star$ and $\mathcal{B}\widetilde{\boldsymbol{u}}$}
We compute the state evolution of the lifted linear system \eqref{eqn:lifted_space_dyn_final} when random control inputs $U^\star$ whose values lie in $[-30,30]$ are applied to it. The obtained states are then used to compute the state dependent matrix $\mathcal{B}$. The control input $\widetilde{\boldsymbol{u}}$ corresponds to the solution to the least square problem given in \eqref{prob:least}. We take $N_1=N_2=15$. It is observed that the relative error between $\widetilde{\mathcal{B}}\boldsymbol{U}^\star$ and $\mathcal{B}\widetilde{\boldsymbol{u}}$ is approximately $3\%$. When $N_1=N_2=25$, the relative error becomes approximately $1\%$.

{\subsection{Comparison with prior method}
We compare our approach with a recent Koopman based method for modeling quadrotor dynamics on $SE(3)$ presented in~\cite{abraham2019active_data_quad_dyn}, in which a data-driven approach was used to approximate the matrices $\mathcal{A}$ and $\mathcal{B}$ based on the following 18 observable functions:}
${z(x)=\left[a_g, \boldsymbol{\omega}, \boldsymbol{v}, g(\boldsymbol{v},\boldsymbol{\omega})\right]^\mathrm{T} \in \mathbb{R}^{18}}$, 
{where $a_g$ is the gravity vector and $g(\boldsymbol{v},\boldsymbol{\omega})=[\boldsymbol{v}_{3} \boldsymbol{\omega}_{3},$ $ \boldsymbol{v}_{2} \boldsymbol{\omega}_{3}, \boldsymbol{v}_{3} \boldsymbol{\omega}_{1}, \boldsymbol{v}_{1} \boldsymbol{\omega}_{3}, \boldsymbol{v}_{2} \boldsymbol{\omega}_{1},\boldsymbol{\omega}_{2} \boldsymbol{\omega}_{3},  \boldsymbol{\omega}_{1} \boldsymbol{\omega}_{3}, \boldsymbol{\omega}_{1} \boldsymbol{\omega}_{2}]$, where $\boldsymbol{v}_i=(\boldsymbol{v})_i$ and $\boldsymbol{\omega}_i=(\boldsymbol{\omega})_i$. In contrast to \cite{abraham2019active_data_quad_dyn}, our approach does not require any data and the observable functions were not guessed. From Table \ref{tab:error},  it can be observed that the approximation error obtained with our approach for $N_1=N_2=25$ is one order of magnitude less than  \cite{abraham2019active_data_quad_dyn} at $t=60s$. However, the approximation error using \cite{abraham2019active_data_quad_dyn} was slightly better than our approach when $N_1=N_2=15$. }
\begin{table}[ht]
\centering
    \centering
    \begin{tabular}{|c|c|c|c|}
       \hline Approx. error & $N_1=N_2=25$ & $N_1=N_2=15$ &  Method \cite{abraham2019active_data_quad_dyn}  \\
\hline $\boldsymbol{x}$ & $4.923\times 10^{-3}$ & $1.697\times 10^{-2}$  & $1.564\times 10^{-2}$ \\
\hline $\boldsymbol{v}$ & $4.167\times 10^{-3}$ & $1.893\times 10^{-2}$  & $1.714\times 10^{-2}$ \\
\hline $[\phi\;\; \theta \;\; \psi]^\mathrm{T}$ & $5.247\times 10^{-4}$  & $2.457\times 10^{-3}$  & $2.243\times 10^{-3}$ \\
\hline
    \end{tabular}
    \caption{{Comparisons with prior methods demonstrate an
order-of-magnitude improvement in approximation error.}}
    \label{tab:error}
    \end{table}

\section{Conclusion\label{sec:conclusions}}
In this paper, we used the framework of Koopman operator to describe the nonlinear dynamics of a quadrotor on $SE(3)$ by means of a linear state space model evolving on the lifted space. We proposed a systematic way to derive a sequence of observable functions that span the lifted space and proved that the latter sequence converges pointwise to the zero function.
This result allowed us to choose a finite subset of this set of functions to form a truncated (approximation of the)  lifted-space. Our simulations indicated that as the dimension of the lifted space dynamics increases, the approximation error decreases. In our future work, we plan to use the derived lifted space linear model for design of controllers for quadrotors.

 \bibliography{main.bib}

\end{document}